\documentclass[11pt,letter]{article}
\usepackage{epsfig,graphicx,subfigure}

\usepackage[english]{babel}
\usepackage{times}
 \usepackage[margin=1in]{geometry}
\usepackage{amsmath, amsthm, amssymb}

\newtheorem{theorem}{Theorem}[section] 
\newtheorem{proposition}[theorem]{Proposition} 
\newtheorem{lemma}[theorem]{Lemma}
\newtheorem{corollary}[theorem]{Corollary}
\newtheorem{definition}[theorem]{Definition}
\newtheorem{remark}[theorem]{Remark}

\usepackage[algoruled,vlined,english]{algorithm2e} 
\usepackage{xspace} 

\usepackage{tikz} 
\usetikzlibrary{calc} 
\usetikzlibrary{automata} 
\tikzset{every state/.style={minimum size=0pt}}

\newcommand{\ujourney}[4][]{\ensuremath{#2\mathop{\leadsto}\limits^{#3}_{#1}#4}}
\newcommand{\macro}[2]{\providecommand{#1}{{\ensuremath{#2}}\xspace}}
\macro{\N}{\mathbb N}
\macro{\Z}{\mathbb Z}
\macro{\R}{\mathbb R}
\macro{\Q}{\mathbb Q}
\macro{\T}{\mathbb T}
\macro{\lifetime}{\cal T}
\macro{\G}{\cal G}
\macro{\TVG}{(V,E,\lifetime,\rho,\zeta)}
\macro{\TVGA}{\cal A}
\macro{\JSET}{{\cal J}^*}
\macro{\J}{\cal J}
\macro{\lab}{}

\title{\bf Expressivity of Time-Varying Graphs and \\
the Power of Waiting   in Dynamic Networks}
\author{A. Casteigts$^\dagger$, P. Flocchini$^\dagger$, E. Godard$^\ddagger$, N. Santoro$^\S$, M. Yamashita$^\P$\medskip\\
  \small
  $^\dagger$ University of Ottawa, Canada.
  \vspace{-4pt}\\\footnotesize
  {\tt  \{casteig,flocchin\}@eecs.uottawa.ca} \\
  \small
  $^\ddagger$ Universit\'e de Provence, Marseille, France.
  \vspace{-4pt}\\\footnotesize
  {\tt  egodard@cmi.univ-mrs.fr}\\
  \small
  $^\S$ Carleton University, Ottawa, Canada.
  \vspace{-4pt}\\\footnotesize
  {\tt santoro@scs.carleton.ca}\\
  \small
  $^\P$ Kyushu University, Fukuoka, Japan.
  \vspace{-4pt}\\\footnotesize
  {\tt  mak@csce.kyushu\nobreakdash-u.ac.jp}
}
\date{}
\begin{document}

\maketitle
\begin{abstract}
In infrastructure-less highly  dynamic networks, computing  and
performing even basic tasks (such as routing and broadcasting) is a
very challenging activity due to the fact that connectivity does not
necessarily hold, and the network may actually be disconnected at
every time instant.  
Clearly the task of designing protocols for these networks is less
difficult   if  the environment allows {\em waiting} (i.e., it
provides the nodes  with store-carry-forward-like mechanisms  such as
local buffering) than if waiting is not feasible. 
No quantitative corroborations of this fact exist (e.g., no answer to
the question: {\em how much easier?}). 
In this paper, we consider these qualitative questions about dynamic networks, 
modeled  as  {\em time-varying} (or {\em evolving}) graphs,  where
edges  exist only at some times. 
We examine the difficulty of the environment in terms of the {\em
  expressivity} of the corresponding time-varying graph; that is in
terms of the language generated by the feasible journeys in the
graph. 
 
We prove that the set of languages ${\cal L}_{nowait}$ when no waiting
is allowed contains   all computable  languages. 
On the other end, using algebraic properties of quasi-orders, we prove
that  ${\cal L}_{wait}$ is just the family of {\em regular}
languages. In other words, we prove that, when waiting is no longer
forbidden, the power of the accepting automaton  (difficulty
of the environment) drops drastically from being  
as powerful as a Turing machine, to becoming that of a
Finite-State machine. This (perhaps surprisingly large) gap is a
measure of the computational power of waiting. 

We also study  {\em bounded waiting}; that is when waiting is allowed
at a node only for at most $d$ time units.  We prove the negative
result that ${\cal L}_{wait[d]} = {\cal L}_{nowait}$;  that is, the
expressivity  decreases  only if the  waiting is finite but
unpredictable (i.e., under the control of the protocol designer and
not of the environment). 
 \end{abstract}

\section{Introduction}

\subsection{Highly Dynamic Networks}
%

Computing in static networks (complex or otherwise) is a subject which has been intensively studied from
many point of views (serial/distributed, centralized/decentralized, offline/online, etc.), and it is one
of the central themes of distributed computing.
Computing in {\em dynamic} networks, that is where the structure of the network
changes in time,  is relatively less understood. Extensive research has been devoted to systems where the
 network dynamics  are due to {\em faults} (e.g., node or edge deletions or additions).
 Indeed fault-tolerance is probably the most profound
concern in distributed computing. Faults however are limited in scope, bounded in number,
and are considered anomalies with respect to the correct behaviour of the system.
The study of 
computing in systems where the network faults are actually extensive and possibly unbounded   
is at the core of the field of self-stabilization; the goal of the research is to devise protocols
that, operating in such extreme faulty conditions, are nevertheless able  to provide correct solutions
if the system instability subsides (for long enough time). Also in this case,  faults in the network structure
 are considered anomalies with respect to the correct behaviour of the system.

What about systems where the instability never ends? where the network is never connected?
where changes are unbounded and occur
continuously? where they are not anomalies but integral part of the nature
of the system?

Such highly dynamic systems do exist, are actually quite widespread, and becoming more ubiquitous.
The most obvious  class is that of wireless ad hoc mobile networks: the topology of the communication network, 
formed by having an edge between two entities when they are in communication range, changes continuously
in time as the movement of the entities destroys old connections and creates new ones. These
changes can be dramatic; {connectivity} does not necessarily hold, at least with the usual meaning of {contemporaneous end-to-end multi-hop paths} between any pair of nodes, and 
 the network may actually be disconnected at every time instant. 
These infrastructure-less highly  dynamic networks,
variously called  {\em  delay-tolerant}, {\em disruptive-tolerant},  {\em challenged}, {\em opportunistic},
have been long and extensively investigated   by the engineering community
and, more recently,  by distributed computing researchers,
especially with regards to the problems of broadcast and routing
 ({\it e.g.} \cite{JLW07,LW09a,Zha06}).

 In these networks, the protocol designer has no a priori knowledge nor control over the trajectories of the
 entities. However,  similar highly dynamic conditions occur also when the mobility of the entities 
 follows a predictable pattern, e.g. periodic or cyclic routes, like in the case of
 public transports with fixed timetables, low earth orbiting (LEO) satellite systems, security guards' tours, etc.
(e.g., see   \cite{LW09b,ZKL+07}).
Interestingly, similar  complex dynamics occur also in environments where there is no
mobility at all, e.g., in  
{\em social networks} (e.g. 
\cite{CLP11,KosKW08}).

Note that when dealing with these dynamic networks,  most of the
basic network and graph concepts - such as path, distance, diameter, connected components, etc -
have no meaning without a temporal
context; indeed, all  the usual connectivity concepts have to be extended  to a temporal version
to take into account the realities of the environments being modeled.

\subsection{Journey and Wait}

     From a formal point of view, the highly dynamic features of  these networks and their temporal nature
    are captured by the model  of {\em time-varying}   graphs (or {\em evolving} graphs), where edges between 
nodes exist only at some times, a priori unknown to the algorithm designer
 (e.g., see \cite{BauCF09,CCF09,CFQS11,CleMPS09,Fer04,FMS09,GAS11,IW11,KLO10,KMO11}).

 A crucial aspect of dynamic networks,  and obviously of time-varying graphs,
is that
a path from a node to another might still exist over time,
even though at no time the path exists in its entirety. 
It is this fact that  renders  routing, broadcasting, and thus computing possible 
in spite of the otherwise unsurmountable  difficulties imposed by the nature of those networks.
 Hence, the notion of  ``path over time", formally 
called {\em journey},  is a fundamental  concepts  and plays a central role
in the definition of almost all concepts related to connectivity in time-varying graphs.

Examined extensively, under a variety of names (e.g., temporal path, schedule-conforming path, time-respecting path, trail), informally a journey
is a walk  $<$$e_1, e_2, ..., e_k$$>$  and a sequence of time instants $<$$t_1, t_2, ..., t_k$$>$ where
 edge $e_i$ exists at time $t_i$  and its latency $\zeta_i$ at that time is such that
$t_{i+1} \geq  t_i + \zeta_i$.

While the concept of journey captures the notion of ``path over time" so crucial in dynamical systems,
it does not yet capture additional  limitations that some of these environments can impose on
the use of the journeys during a computation.
More specifically, there are systems that provide  the entities with store-carry-forward-like mechanisms
(e.g., local buffering); thus an entity wanting to communicate with a specific other entity at time $t_0$,
can wait until the opportunity of communication presents itself. 
There are however environments where such
a provision is not available (e.g., there are no buffering facilities), and thus waiting is not allowed.
  In time-varying graphs, this distinction is the one between a {\em direct}  journey where 
  $\forall i,  t_{i+1} =  t_i + \zeta_i$,
   and  an
  {\em indirect} journey where  $\exists i,  t_{i+1} >  t_i + \zeta_i$.

With regards to
problem solving,
 any restriction, imposed  by the nature of the
system on the protocol designer,
has clearly an impact on the computability and complexity
of problems. In dynamic networks,
computing (already a difficult task) is intuitively more
difficult  in  environments  that do not allow waiting 
than in
those   where waiting is possible; that is,  
environments where the only feasible  journeys are the direct ones 
are clearly more challenging (for the problem solver) than those where
journeys can be indirect.

In the common view of 
the environment as the adversary 
that the problem solver has to face, an environment that forbids waiting
is clearly a more difficult (i.e. stronger) adversary than the one that allows waiting.
The natural and immediate question is {\em ``how much stronger is the adversary if
waiting is not allowed?"}
 which can be re-expressed as: {\em ``if waiting is allowed, how much easier is to solve problems?"}, or simply {\em ``what is the computational power of waiting?"}


A first difficulty in addressing these important  questions is 
that most of the terms are qualitative, and  currently  there are no 
 measures 
that allow to  quantify  even the main  concepts  e.g. 
``complexity" of the environment,  ``strength" of the adversary,
``difficulty" of solving problems, etc.

 In this paper,  motivated by these qualitative questions, we 
 examine the complexity of  the environment (modeled as a time-varying graph)
 in terms of its
{\em expressivity},
 and establish results showing the  
 (surprisingly dramatic)  difference  that the possibility of waiting creates.

\subsection{Contributions}

Given a dynamic network modeled as a time-varying graph $\G$, 
a journey  in $\G$ can be viewed as a word on the alphabet of the edge labels;
in this light,
the class of feasible journeys defines the language $L_f(\G)$ expressed by $\G$,
 where $f\in\{wait, nowait\}$ indicates whether or not  indirect journeys are considered
  feasible by the environment.

  We focus on the sets of languages  ${\cal L}_{nowait} =\{L_{nowait} (\G) : \G \in {\cal U}\}$
   and     ${\cal L}_{wait} =\{L_{wait} (\G) : \G \in {\cal U}\}$, where $\cal U$ is the set of all time-varying graphs; that is, we look at the languages expressed when waiting is, or is not allowed.
For each of these two sets,
the complexity of  recognizing any language in the set (that is,  the computational power
needed by the accepting automaton)  defines  the level of
difficulty of the environment.

We first study the expressivity of  time-varying graphs when waiting is not allowed,
that is the only feasible journeys are direct ones. We prove that 
the set  ${\cal L}_{nowait}$ contains  all  computable  languages.
That is, we show that,  for any computable  language $L$,
 there exists a time-varying graph $\G$ such that $L= L_{nowait}(\G)$. 

We next examine  the expressivity of  time-varying graphs if indirect journey are allowed.
We prove that  ${\cal L}_{wait}$ is precisely the set of {\em regular} languages.
 The proof is algebraic and based on order techniques, relying on a theorem by 
 Harju and Ilie~\cite{harju}  that enables to characterize regularity from the closure of 
 the sets from a well quasi-order.
In other words,  we prove that, when waiting is no longer forbidden, the power of the accepting automaton  
 (i.e., the difficulty of the environment, the power of the adversary), drops drastically from being
 as powerful as a Turing machine, 
 to becoming that of a Finite-State Machine. This 
 (perhaps surprisingly large) gap is a measure of
the computational power of waiting.

To better understand the power of waiting, 
we then turn our attention to  {\em bounded waiting}; that is when 
indirect journeys are considered feasible if the pause 
between consecutive edges in the journeys have a bounded duration $d>0$.
In other words, at each step of the journey, waiting is allowed only for at most $d$ time units. We 
examine  the set  ${\cal L}_{wait[d]}$ of  the languages expressed by time-varying graphs
when waiting is allowed up to  $d$ time units. We prove the negative result that for any
fixed $d\geq 0$,
${\cal L}_{wait[d]} = {\cal L}_{nowait}$, which implies that the complexity of
the environment is not affected by allowing waiting for a limited amount of time.
As a result, the power of the adversary is decreased only if it has no control over the
length of waiting, i.e., if the waiting is unpredictable. 
%

\subsection{Related Work}

The literature on dynamic networks and dynamic graphs could fill volumes. Here  we briefly
mention only some of the work most directly connected to the results of this paper.
 The idea of representing dynamic graphs as a sequence of (static) graphs, called
  {\em evolving graph}, 
  was introduced in~\cite{Fer04}, 
  to study basic network problems 
 in dynamic networks from
 a centralized point of view \cite{BhF03,BFJ03}.
 The evolving graph views the dynamics of the system as a sequence
of {\em global} snapshots (taken either in discrete steps or when events occur).
 The  equivalent model  of {\em time-varying graph} (TVG), 
 formalized in \cite{CFQS11} and used here, 
 views the dynamics of the system from the {\em local} point of view of the entities:
for any given entity, the local edges and neighborhood can be considered independently from the entire graph (e.g. how long it is available, with what properties, with what latency, etc.).

Both viewpoints
 have been extensively employed in the analysis of 
 basic 
 problems  such as routing, broadcasting, gossiping and other forms of information spreading
 (e.g.,   \cite{AKL08,CFMS10,CFMS11,CoN10,OW05});
  to study problems of exploration in vehicular networks with periodic routes  \cite{FMS09,IW11};
to examine failure detectors \cite{GAS11} and consensus \cite{KLO10,KMO11};
 for the probabilistic analysis of informations spreading
(e.g.,   \cite{BauCF09,CleMMPS08}); 
and 
    in the investigations of emerging properties in
 social networks (e.g., \cite{KeK02,TSM+09}).
 A characterization of classes of TVGs with respect to properties typically assumed in
the research can be found in \cite{CFQS11}.
The related investigations on dynamic networks include  also the extensive work on {\em population protocols}
(e.g., \cite{AAER07, CMS09});
interestingly, the setting over which population protocols are
defined is a particular class of time-varying graphs (recurrent interactions over a connected underlying graph).
The impact of bounded waiting in dynamic networks has been investigated 
for exploration~\cite{IW11}.

\section{Definitions and Terminology}

 \paragraph{\bf Time-varying graphs:}
A {\em time-varying graph}\footnote{We use the notation for time-varying graphs introduced in \cite{CFQS11}}
  \G  is a quintuple  \G = \TVG, where
 $V$ is a finite set of entities  or {\em nodes};  $E \subseteq V \times V \times \Sigma$ is a finite set of relations between these entities ({\em edges}), possibly labeled by symbols in an alphabet $\Sigma$. The system is studied over a given time span $\lifetime \subseteq \T$ called {\em lifetime}, where $\T$ is the temporal domain (typically,  $\N$ or $ \R^+$ for discrete and continuous-time systems, respectively); $\rho: E \times \lifetime \to \{0,1\}$ is the {\em presence} function, which indicates whether a given edge is available at a given time;
$\zeta: E \times \lifetime \to \T$, is the {\em latency} function, which indicates the time it takes to cross a given edge if starting at a given date (the latency of an edge could vary in time).
Both presence and latency are arbitrary {\em computable} functions.
The directed edge-labeled graph $G=(V,E)$,  called the {\em footprint} of \G,  may contain loops, and
it may have more than one edge between the same nodes, but all with different labels.

A path over time, or {\em journey},
is a sequence     $<(e_1,t_1), (e_2,t_2) , ..., (e_k, t_k)>$   where  $<e_1, e_2, ..., e_k>$
is a walk in the footprint G, $\rho(e_i,t_i)=1$ (for $1\leq i <k$),  and  $\zeta(e_i,t_i)$  is such that
$t_{i+1} \geq  t_i + \zeta(e_i,t_i)$ (for $1\leq i <k$). 
If   $\forall i,  t_{i+1} =  t_i + \zeta(e_i,t_i)$ the journey is said to be  {\em direct}, 
  {\em indirect} otherwise.  We denote by $\J^*(\G)$ the set of all
journeys in \G.

The time-varying-graph (TVG)  formalism can arguably describe a multitude of different scenarios, from transportation networks to communication networks, complex systems, or social networks \cite{CFQS11}. 
Figure~\ref{fig:examples} shows two simple examples of TVGs, depicting respectively a transportation network (Figure~\ref{fig:example1}) and a communication network (Figure~\ref{fig:example2}). In the transportation network, an edge from node $u$ to node $v$ represents the possibility for some agent to move from $u$ to $v$; typical edges in this scenario are available on a {\em punctual} basis, {\it i.e.,} the presence function $\rho$ for these edges returns $1$ only at particular date(s) when the trip can be started. The latency function $\zeta$ may also vary from one edge to another, as well as for different availability dates of a same given edge (e.g. variable traffic on the road, depending on the departure time). In the communication network, the labels are not indicated; shown instead are the intervals of time when the presence function $\rho$ is 1. An example of direct journey in this graph is $\J_1=\{(ab,2),(bc,2+\zeta)\}$. Examples of indirect ones
  include $\J_{2}=\{(ac,2),(cd,5)\}$, and  $\J_{3}=\{(ab,2),(bc,2+\zeta),(cd,5)\}$. 


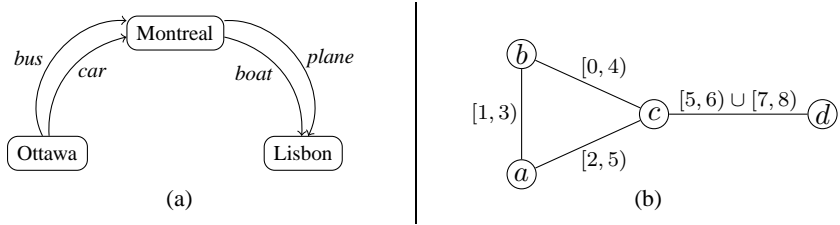
\begin{figure}[h]
  \centering
  \begin{tabular}{c|c}
    \subfigure[]{
      \label{fig:example1}
      \quad
      \begin{tikzpicture}[yscale=1.6, xscale=1.7]
        \tikzstyle{every node}=[draw, rectangle, rounded corners=4pt, font=\scriptsize]
        \path (0,0) node (o) {Ottawa};
        \path (1,1) node (m) {Montreal};
        \path (2,0) node (l) {Lisbon};
        \tikzstyle{every path}=[->]
        \tikzstyle{every node}=[font=\scriptsize]
        \draw (o) edge[bend left=60] node[left] {{\it bus}} (m);
        \draw (o) edge[bend left=40] node[right] {\it car} (m);
        \draw (m) edge[bend left=40] node[left] {{\it boat}} (l);
        \draw (m) edge[bend left=60] node[right] {\it plane} (l);
        \path (o)+(0,-.3) coordinate;
      \end{tikzpicture}
      \quad
    }
    &
    \subfigure[]{
      \label{fig:example2}
      \quad
  \begin{tikzpicture}[scale=1.6]
    \tikzstyle{every node}=[draw, circle, minimum size=11pt, inner sep=0pt]
    \path (0,0) node (a){$a$};
    \path (a)+(0, 1) node (b){$b$};
    \path (a)+(1.1,.5) node (c){$c$};
    \path (c)+(1.4,0) node (d){$d$};
    \tikzstyle{every node}=[font=\scriptsize, inner sep=1pt]
    \draw (a)--node[midway, left]{$[1,3)$}(b);
    \draw (a)--node[midway, xshift=-4pt, yshift=-1pt, below right]{$[2,5)$}(c);
    \draw (b)--node[midway, xshift=-4pt, yshift=1pt, above right]{$[0,4)$}(c);
    \draw (c)--node[midway, above]{$[5,6)\cup [7,8)$}(d);
  \end{tikzpicture}
    \quad
    }
  \end{tabular}
  \caption{\label{fig:examples} \footnotesize Two examples of time-varying graphs, highlighting 
 (a)  the labels, and (b) the presence function.}
\end{figure}

 \paragraph{\bf TVG-automata:}
Given a time-varying graph $\G = \TVG$ whose edges are labeled over  $\Sigma$, 
we define a TVG-automaton
 $\TVGA({\G})$ as the 5-tuple $\TVGA({\G})=(\Sigma, S, I, {\cal E}, F)$ where
$\Sigma$ is the input {\em alphabet};
$S=V$ is the set of {\em states};
$I \subseteq S$ is the set of {\em initial states}; $F \subseteq S$ is the set of  {\em accepting
states};
${\cal E}\subseteq S \times \lifetime \times \Sigma \times S \times \lifetime$ is the set of {\em transitions} such that $(s,t,a,s',t') \in {\cal E}$
 iff $\exists e=(s,s',a)\in E : \rho(e,t)=1,\zeta(e,t)=t'-t$. In the following we
 shall denote  $(s,t,a,s',t') \in {\cal E}$ also by $s,t \overset{a}{\to} s',t'$.
A  TVG-automaton $\TVGA({\G})$ is {\em deterministic} if 
 for any time $t\in\lifetime$, any   state $s\in S$, any symbol $a\in\Sigma$, there is
 at most one transition of the form $(s, t \overset{a}{\to} s',t')$; it is {\em non-deterministic} otherwise.

Given a TVG-automaton $\TVGA({\G})$, a {\em journey  in} $\TVGA({\G})$ is a finite sequence of transitions 
${\cal J}=(s_{0},t_{0} \overset{a_0}{\to} s_{1},t_{1}),(s_{1},t_{1}' \overset{a_1}{\to} s_{2},t_{2}) \dots (s_{p-1},t_{p-1}' \overset{a_{p-1}}{\to} s_{p},t_{p})$ such that
the sequence $<$$(e_0,t_{0}), (e_1,t'_{1}),...,(e_{p-1},t'_{p-1})$$>$
is a journey in  \G  and
$t_p = t'_{p-1}+\zeta(e_{p-1},  t'_{p-1})$, where 
$e_i=(s_i,s_{i+1},a_i)$ (for $0\leq i < p$).
Consistently with the above definitions, we say that ${\cal J}$ is {\em direct} if $\forall i, t_i'=t_i$ (there is no pause between transitions), and {\em indirect} otherwise ({\it i.e.,} $\exists i : t_i'>t_i$). 
We denote by $\lambda(\J)$ the associated word $a_0,a_1,...a_{p-1}$ and by
$start(\J)$ and $arrival(\J)$ the dates $t_0$ and $t_p$, respectively.
To complete the definition, an  {\em empty} journey $\J_{\emptyset}$ consists of
 a single state, involves no transitions,   its associated word is the empty 
 word $\lambda(\J_{\emptyset})=\varepsilon$, and its arrival date is the starting date.

A journey is said {\em accepting} iff it starts in an initial state $s_0 \in I$ and ends in a accepting state $s_p\in F$.
 A TVG-automaton $\TVGA({\G})$ {\em accepts} a word $w\in\Sigma^*$ iff there exists an accepting journey  ${\cal J}$ such that $\lambda(\J)=w$. 

 Let $L_{nowait}(\G)$ denote
 the set of words (i.e., the {\em language}) accepted by TVG-automaton $\TVGA({\G})$  using only direct journeys, and let
 $L_{wait}(\G)$ be the language recognized if journeys are allowed to be indirect. 
 Given the set $\cal U$ of all possible TVGs, let us denote   ${\cal L}_{nowait} =\{L_{nowait} (\G) : \G \in {\cal U}\}$
   and     ${\cal L}_{wait} =\{L_{wait} (\G) : \G \in {\cal U}\}$ the sets of all languages being possibly accepted by a TVG-automaton  if journeys are constrained to be direct (i.e., no waiting is allowed)
    and if they are unconstrained (i.e., waiting is allowed), respectively. 

In the following, when no ambiguity arises, we will use  interchangeably
the terms node and state, and the terms edge and transition; the term journey 
will be used both in reference to the sequence of edges in the TVG and to  the corresponding
sequence of transitions in the associated TVG-automaton.

The closest concept to  TVG-automata are  {\em Timed automata} proposed by \cite{TA} to model real-time systems. 
A timed automaton has real valued clocks and the transitions are guarded with simple
comparisons on the clock values; 
 with only one clock 
and no reset it is a TVG-automaton with 0 latency. 

\paragraph{\bf Example of TVG-automaton:}
Figure~\ref{fig:anbn}  shows an example   of a deterministic TVG-automaton that recognizes the context-free language $a^n b^n$ for $n \ge 1$ (using only direct journeys). 
Consider the graph $\G_1=\TVG$,  composed of three nodes:
 $V=\{v_0,v_1,v_2\}$,
and five  edges: 
$E=\{(v_0,v_0,a)$, $(v_0,v_1,b)$, $(v_1,v_1,b)$, $(v_1,v_2,b)$,  $(v_0,v_2,b))\}$. 
 The presence and latency functions are as shown  in Table \ref{functions},
 where $p$ and $q$ are two distinct  prime numbers greater than 1.
 Consider now the corresponding automaton $\TVGA({\G_1})$ 
 where  $v_0$ is the initial state and $v_2$ is the accepting state. For clarity, let us assume that $\TVGA(\G_1)$ starts at time $1$ (the same behavior could be obtained by modifying slightly the formulas involving $t$ in Table~\ref{functions}).
It is clear that the $a^n$ portion of the word $a^n b^n$     is read entirely at $v_0$ within $t=p^n$ time. 
If $n=1$, at this time  the only available edge is $e_3$ (labeled {\tt b}) which allows to correctly accept $ab$. Otherwise ($n>1$)  at   time $t=p^n$, the 
only available edge is $e_1$ which allows to start reading the $b^n$ portion of the word. 
By construction of $\rho$ and $\zeta$, edge $e_2$ is always present except for the very last $b$, which has to be 
read at time $t=p^nq^{n-1}$.  At that  time, only $e_4$ is present and the word is correctly recognized. It is easy to verify that only these words are recognized, and the automaton is deterministic.
The reader may have noticed the basic principle employed here (and  later in the paper)  of using
latencies as a means to {\em encode} words into time, and presences as a means to {\em select} through opening
the appropriate edges at the appropriate time.

\begin{table}[tbh]
\begin{center}
\begin{tabular}{|c|c|c|c|c|}
\hline
$e$ && $\rho(e,t) =1$ iff &&$\zeta(e,t)\ =$\\
\hline
&& &&\\
\hline
$e_0$ &&  always true & &  $(p-1)t$ \\
\hline
$e_1$ &&       $t>p$ & & $(q-1)t$ \\
\hline
$e_2$&&       $t\ne p^iq^{i-1}$,$i>1$ && $(q-1)t$ \\
\hline
$e_3$&&       $t=p$ && any \\
\hline
$e_4$&&      $t=p^iq^{i-1}$, $i>1$ &&any   \\
\hline
\end{tabular}
\caption{Presence and Latency functions for TVG-automaton $\G_1$.}
\label{functions}
 \end{center}
\end{table}

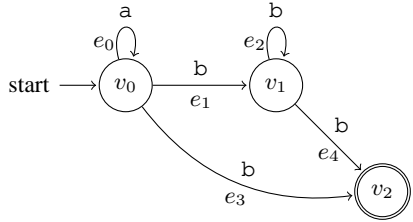
\begin{figure}[h]
 \centering
 \begin{tikzpicture}[shorten >=1pt,node distance=2cm,auto, font=\footnotesize]
   \node[state,initial] (q0) {$v_0$};
   \node[state](q1) [right of= q0] {$v_1$};
   \node[state,accepting](q2) [below right of= q1] {$v_2$};
   
   \path[->] (q0) edge [loop above] node[left,yshift=-6pt] {$e_0$~} node {\tt a} ()
   edge node[below]{$e_1$} node {{\tt b}} (q1)
   (q1) edge [loop above] node[left,yshift=-6pt] {$e_2$~} node {{\tt b}} ()
   edge node[below]{$e_4$} node[inner sep=2pt] {{\tt b}} (q2);
   \path[->] (q0) edge[bend right] node[below]{$e_3$} node[inner sep=2pt] {{\tt b}} (q2);
 \end{tikzpicture}
 \caption{\label{fig:anbn}A TVG-automaton $\G_1$ such that $L_{nowait}(\G_1)=\{a^nb^n : n\ge 1\}$.}
\end{figure}



%

\section{No waiting allowed}
\label{sec:no-waiting}

This section focuses on the expressivity of time-varying graphs
 when only {\em direct} journeys are allowed.
 We prove that 
      ${\cal L}_{nowait}$ includes  all  computable languages.


%

Let $L$ be an arbitrary  computable language defined over a finite alphabet 
$\Sigma$. Let $\varepsilon$ denote the empty word; note that $L$ might or might not contain
$\varepsilon$.
The notation $\alpha.\beta$ indicates the concatenation of $\alpha\in
\Sigma^*$ with $\beta\in\Sigma^*$.

Let $q=|\Sigma|$ be the size of the alphabet, and
w.l.o.g  assume that $\Sigma=\{0, \dots, q-1\}$. 
We define an injective encoding $\varphi: \Sigma^*{\to} \N $  associating to each 
word $w=a_0.a_1\dots a_k \in\Sigma^*$  the sum
$q^{k+1}+\mathop\sum_{j=0}^{k} a_j q^{k-j}$. 
It is exactly the integer corresponding to $1.w$ interpreted in
base $q$.
By convention, $\varphi(\varepsilon) = 0$.
 
Consider now  the TVG
 $\G_2$ where $V=\{v_0,v_1,v_2\}$,
$E=\{\{ (v_0,v_1,i),i\in\Sigma\} \cup \{\{ (v_0,v_2,i),i\in\Sigma\}, \cup \{(v_1,v_1,i),i\in\Sigma\}\cup\{(v_1,v_2,i),i\in\Sigma\}\cup\{(v_2,v_1,i),i\in\Sigma\}\cup\{(v_2,v_2,i),i\in\Sigma\}\}$.
The   presence and latency functions 
are defined relative to which node is 
the end-point of an edge.  
For all $u\in\{v_0,v_1,v_2\}$,  $i\in\Sigma$, and $t\geq 0$, we  define 
\begin{itemize}
\item  $\rho((u,v_1,i),t) = 1$ iff $t\in\varphi(\Sigma^*)$ and $\varphi^{-1}(t).i\in L,$
\item  $\zeta((u,v_1,i),t) = \varphi(\varphi^{-1}(t).i)-t$

  \item $\rho((u,v_2,i),t) = 1$ iff $t\in\varphi(\Sigma^*)$ and $\varphi^{-1}(t).i\notin L,$
  \item $\zeta((u,v_2,i),t) = \varphi(\varphi^{-1}(t).i)-t$
\end{itemize}

Consider the corresponding TVG-automaton
$ \TVGA(\G_2)$ where
 $v_0$ is the initial state, 
and  $v_1$ is the unique  accepting state  if the $\varepsilon\notin L$ (see Figure~\ref{fig:nowaiting}), while  
both    $v_0$ and $v_1$ are  accepting states if $\varepsilon \in L$.

\begin{figure}[h]
  \centering
   \begin{tikzpicture}
     \tikzstyle{every node}=[node distance=3cm,minimum size=0pt];
     \node[initial, state] (v0) {$v_0$};
     \node[accepting, state, above right of=v0,node distance=2.2cm] (v1) {$v_1$};
     \node[state, right of=v1] (v2) {$v_2$};
     \tikzstyle{every node}=[above, font=\scriptsize, inner sep=1pt]
     \path (v1) edge[loop below, distance=1.4cm, in=128, out=52, ->] node {\tt q-1} (v1);
     \path (v1) edge[loop above, distance=.4cm, in=110, out=70, ->] node (v1v1){\tt 0} (v1);
     \path (v2) edge[loop below, distance=1.4cm, in=128, out=52, ->] node {\tt q-1} (v2);
     \path (v2) edge[loop above, distance=.4cm, in=110, out=70, ->] node (v2v2){\tt 0} (v2);
     \path (v0) edge[bend left=14,->] node[above left] {\tt q-1} (v1);
     \path (v0) edge[bend right=14,->] coordinate[above] (v0v1) node[below right] {\tt 0} (v1);
     \path (v0) edge[bend right=30,->] node[above left] {\tt q-1} (v2);
     \path (v0) edge[bend right=44,->] coordinate[above] (v0v2) node[below right] {\tt 0} (v2);
     \path (v1) edge[bend right=15,->] node[above] {\tt q-1} (v2);
     \path (v1) edge[bend right=36,->] node[above] (v1v2) {\tt 0} (v2);
     \path (v1) edge[bend left=38,<-] node[above,inner sep=1pt] {\tt q-1} (v2);
     \path (v1) edge[bend left=15,<-] node[above,inner sep=1pt] (v2v1) {\tt 0} (v2);
     \tikzstyle{every node}=[font=\scriptsize]
     \path (v0v1) node[above, xshift=-2.4pt] {$\dots$};
     \path (v0v2) node[above=.2pt, xshift=-5pt] {~$\dots$};
     \path (v1v2) node[above=.2pt] {~$\dots$};
     \path (v2v1) node[above=.5pt] {~$\dots$};
     \path (v1v1) node[above=2pt] {~$\dots$};
     \path (v2v2) node[above=2pt] {~$\dots$};
   \end{tikzpicture}
 \caption{\label{fig:nowaiting}A TVG $\G_2$ that recognizes an arbitrary language $L$ (case $\varepsilon \notin L$).}
\end{figure}
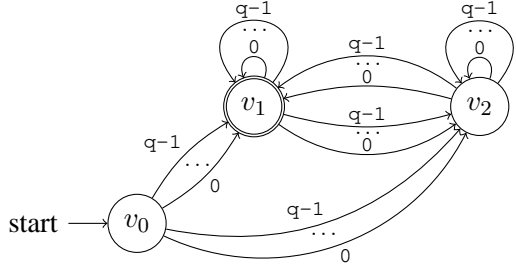

\begin{theorem}
  $L_{nowait}(\G_2)=L$. 
\end{theorem}
\begin{proof}


 We want to show  there is a unique accepting journey $\J $ with  $\lambda(\J)=w$ iff $w \in L$.
   
\noindent  We first show that for all words $w \in \Sigma^*$, there is
exactly one direct journey $\J$ in $\TVGA(\G_2)$ such that
$\lambda(\J)=w$, and in this case $arrival(\J)=\varphi(w)$. 
  This is proven by induction on $k\in\N$, the length of the words. 
It clearly holds for $k=0$ since the only word of that length is $\varepsilon$ and  $\varphi(\varepsilon)=0$.  
   Let $k\in\N$. Suppose now that for all 
  $w \in\Sigma^*, |w|=k$ we have  exactly one associated direct journey, and
  $arrival(\J)=\varphi(w)$. 

 \noindent    
 Consider $w_1\in\Sigma^*$ with $|w_1|=k+1$. Without loss of generality, let  $w_1=w.i$
 where $w\in\Sigma^*$
  and $i\in\Sigma$.
  By induction there is exactly one direct journey $\J $ with  $\lambda(\J)=w$. 
  Let $u= arrival(\J)$ be the node of arrival and $t$ the arrival time.
    By induction, $t\in\varphi(\Sigma^*)$; furthermore
  since the presence function  depends only on the node of arrival and not  on the node of origin,
  there exists  exactly one transition,
  labeled $i$ from $u$.
 So there exists only one direct journey labeled by $w_1$.
By definition of the latency function, its arrival time is  
$\varphi(\varphi^{-1}(t).i) = \varphi(w.i) = \varphi(w_1)$.
  This ends the induction.

\noindent
We now show that such a unique  journey is accepting iff $w\in L$.
In fact,  by construction of the presence function, every 
journey that corresponds to  $w\in L, w\neq \varepsilon, $ 
ends in  $v_1$, which is an accepting state. The empty journey
corresponding to  $\varepsilon$ ends in $v_0$ which, by definition, is
accepting iff $\varepsilon\in L$.   
\end{proof}

%



\section{Waiting allowed}
\label{sec:waiting}

We now turn the attention to the case of time-varying graphs where
{\em indirect} journeys are possible, that is entities have the choice
to wait for future opportunities of interaction rather than seizing
only those that are directly available.
In striking contrast with the non-waiting case, we show that the
languages ${\cal L}_{wait}$ recognized by TVG-automata consists only
of regular languages.  
Let ${\cal R}$ denote the set of regular languages. 
We show that

\begin{theorem}
\label{theo:wait}
${\cal L}_{wait}  = {\cal R}$.     
\end{theorem}

\begin{proof}(\emph{of inclusion for regular languages})
  This first inclusion follows easily from observing that any finite-state
  machine (FSM) is a particular TVG-automaton whose edges are always
  present and have a nil latency. The fact that we allow waiting here
  does not modify the behavior of the automata as long as we consider
  deterministic FSMs only (which is sufficient), since at most one
  choice exists at each state for each symbol read. Thus, for any
  regular language $L$, there exists a corresponding TVG $\G$ such
  that $L_{wait}(\G)=L$.
\end{proof}


The reverse inclusion is more involved.
Consider $\G=\TVG$ with labels in $\Sigma$, we have to show that $L_{wait}(\G) \in \cal R.$ 

The proof is algebraic, and based on order techniques, relying on a theorem of Harju and Ilie (Theorem~4.16 in~\cite{harju}) that enables to characterize regularity from the closure of the sets from a well quasi-order. We will use  here an inclusion order on journeys (to be defined formally below). Informally, a journey $\J$ is included in another journey $\J'$ if its sequence of transition is included (in the same order) in the sequence of transitions of $\J'$.
  It should be noted that sets of indirect journeys from one node to another are obviously closed 
  under this inclusion order (on the journey $\J$  it is possible to wait on a node as if the missing transitions 
  from $\J'$ were taking place), which is not the case for direct journeys as it is not possible to wait.
 In order to apply the theorem, we have to show that this inclusion order is a well quasi-order,
 i.e. that it is not possible to find an infinite set of journeys such that none of them could be included in another from the same set.

Let us first introduce some definitions and results about quasi-orders.
We denote by $\leq$ a quasi-order over a given set $Q$.
A set $X\subset Q$ is an \emph{antichain} if all elements of $X$
are pairwise incomparable. The quasi-order $\leq$ is \emph{well
  founded} if in $Q$, there is no infinite descending sequence 
$x_1\geq x_2\geq x_3\geq \dots$ (where
$\geq$ is the inverse of $\leq$) such that for no $i$, $x_i\leq
x_{i+1}.$
If $\leq$ is well founded and all antichains are finite then $\leq$
is a \emph{well quasi-order} on $Q$.
When $Q=\Sigma^*$ for alphabet $\Sigma$, a quasi-order is
\emph{monotone} if for all $x,y,w_1,w_2\in \Sigma^*$, we have $x\leq y
\Rightarrow w_1xw_2\leq w_1yw_2$.


    A word $x\in\Sigma^*$ is a
    \emph{subword} of $y\in\Sigma^*$ if $x$ can be obtained by
    deleting some letters on $y$. 
    This defines a relation that is obviously transitive and we denote $\subseteq$ the 
    \emph{subword order} on $\Sigma^*$.
We can extend the $\subseteq$ order to labeled walks as follows:
given two walks $\gamma,\gamma'$  on the footprint $G$ of $\G$,
we note $\gamma\subseteq\gamma'$ if
  $\gamma$ and $\gamma'$ begin on the same node and  
 end on the same node,
 and
   $\gamma$ is a \emph{subwalk} of $\gamma'$.


   Given a date $t\in \lifetime$ and a word $x$ in $\Sigma^*$, we denote by $\JSET{t,x}$ the set $\{\J \in \J^*(\G): start(\J)=t, \lambda(\J)=x\}$. $\JSET{0,x}$ is simply denoted $\JSET{x}$. Given two nodes $u$ and $v$, we allow the notation $\ujourney[t]{u}{x}{v}$ if there exists a journey from $u$ to $v$ in $\JSET{t,x}$.
   Given a journey $\J$, $\overline{\J}$ is the corresponding labeled
  walk (in the footprint $G$). We will denote by $\Gamma(x)$ the set $\{\overline{J} : \lambda(\J)=x\}$.

  Let $x$ and $y$ be two words in $\Sigma^*$. 
  We define the quasi-order $\prec$, as follows:  $x\prec y$ if
  $$\forall \J\in\JSET{y}, \exists\gamma\in\Gamma(x), 
  \gamma\subseteq\overline{\J}.$$
The relation ${\prec}$ is obviously reflexive.
 We now  establish the link between comparable words and
their associated journeys and walks, and state some 
useful properties of relation ${\prec}$.
%

\begin{lemma}\label{walktojourney}
  Let $x,y\in\Sigma^*$ be such that $x\prec y$. Then for any $\J_y\in\JSET{y}$,
  there exists $\J_x\in\JSET{x}$ such that
  $\overline{\J_x}\subseteq\overline{\J_y}$,
   $start(\J_x)=start(\J_y),$
  $arrival(\J_x)=arrival(\J_y).$
\end{lemma}
\begin{proof}
  By definition, there exists a labeled path $\gamma\in\Gamma(x)$ such
  that $\gamma\subseteq \overline{\J_y}$.
  It is then possible to find a journey $\J_x\in\JSET{x}$ with 
  $\overline{\J_x}=\gamma$ and $arrival(\J_x)=arrival(\J_y)$ by using 
  for every edge of $\J_x$ the schedule of the same edge in $\J_y.$ 
\end{proof}

\begin{proposition}
\label{prop:trans}
  The relation $\prec$ is transitive.
\end{proposition}
\begin{proof}
  Suppose we have $x\prec y$ and $y\prec z$. Consider
  $\J\in\JSET{z}$. By Lemma~\ref{walktojourney}, we get a journey
  $\J_y\in\JSET{y}$, such that $\overline{\J_y}\subseteq\overline{\J}$.
  By definition, there exists $\gamma\in\Gamma(x)$ such that
  $\gamma\subseteq\overline{\J_y}.$ Therefore
  $\gamma\subseteq\overline{\J},$ and finally $x\prec z.$
\end{proof}

The main proposition to be proved now is the following
\ 
  \begin{proposition}
  \label{prop:wqo}
  $(\Sigma^*,\prec)$ is a well quasi-order.
  \end{proposition}

Indeed, consider the two following results.

\begin{definition}
  Let $L\subset\Sigma^*$. For any quasi-order $\leq$, 
  we denote $\textsc{Down}_\leq(L)=\{x\mid \exists y\in L, x\leq y\}.$
\end{definition}
  
This is a corollary of Lemma~\ref{walktojourney}    
\begin{corollary}
  \label{lm:conseq}
  Consider the language $L$ of words induced by labels of journeys
  from $u$ to $v$ starting at time $t$. Then $\textsc{Down}_\prec(L)   L.$  
\end{corollary}

The following theorem is due to Harju and Ilie:
\begin{theorem}[\cite{harju}]
\label{theo: harju}
  For any monotone well quasi order $\leq$ of $\Sigma^*,$ for any
  $L\subset\Sigma^*,$ the language $\textsc{Down}_\leq(L)$ is
  regular. 
\end{theorem}

From Proposition \ref{prop:wqo}, Corollary \ref{lm:conseq}, and
Theorem \ref{theo: harju},
 the claim of  Theorem \ref{theo:wait} will immediately follow.
The remaining of this section is devoted to the proof that $\prec$ is a
well quasi-order. We have first to prove the following.  


\begin{proposition}
\label{prop:mon}
  The quasi-order $\prec$ is monotone.
\end{proposition}
\begin{proof}
  Let $x,y$ be such that $x\prec y.$ Let $z\in\Sigma^*$. Let
  $\J\in\JSET{yz}$. Then there exists $\J_y\in\JSET{y}$ and
  $\J_z\in\JSET{start(\J_y),z}$ such that the end node of $\J_y$ is
  the start node of $\J_z$.
  By Lemma~\ref{walktojourney}, there exists $\J_x$ that ends in the
  same node as $\J_y$ and with the same $arrival$ time. 
  We can consider $\J'$ the concatenation of $\J_x$ and $\J_z$. By
  construction $\overline{\J'}\in\Gamma(xz)$, and
  $\overline{\J'}\subseteq\overline{\J}.$
  Therefore $xz\prec yz.$
  The property $zx\prec zy$ is proved similarly using the $start$
  property of Lemma~\ref{walktojourney}.
\end{proof}

\begin{proposition}
\label{prop:wfund}
  The quasi-order $\prec$ is well funded.
\end{proposition}
\begin{proof}
 Consider a descending chain $x_1\succ x_2\succ x_3\succ \dots$ such
 that for no $i$ $x_i\prec x_{i+1}.$ We show that this chain is
 finite. Suppose the contrary.
 By definition of $\prec$, we can find $\gamma_1,\gamma_2,\dots$ such
 that for all $i$,  $\gamma_i\in\overline{\JSET{x_i}},$ and such that
 $\gamma_{i+1}\subseteq\gamma_i.$ This chain of walks is necessarily
 stationary and there exits $i_0$ such that
 $\gamma_{i_0}=\gamma_{i_0+1}.$ 
 Therefore, $x_{i_0}=x_{i_0+1},$ a contradiction. 
\end{proof}

To prove that $\prec$ is a well quasi-order, we now have to prove that
all antichains are finite.
\macro{\leqP}{\leq_{\mathcal P}}
  Let $(Q,\leq)$ be a quasi-order. For all $A,B\subset Q$, we denote
  $A\leqP B$ if there exists an injective mapping
  $\varphi:A\longrightarrow B$, such that for all $a\in A$, $a\leq\varphi(a).$
  The relation \leqP is transitive and defines a quasi-order on
  $\mathcal P(Q).$

About the finiteness of antichains, we recall the following
result

\begin{lemma}[\cite{higman}]\label{higmans}
  Let $(Q,\leq)$ be a well quasi-order. Then $(\mathcal
  P(Q),\leqP)$ is a well quasi-order.
\end{lemma}
  
\noindent and 
 the fundamental result of Higman:
 
\begin{theorem}[\cite{higman}]
 Let $\Sigma$ be a finite alphabet. Then $(\Sigma^*,\subseteq)$ is a
 well quasi-order.
\end{theorem}
  
\noindent This implies that our set of journey-induced walks is
also a well quasi-order for $\subseteq$ as it can be seen as a 
special instance of Higman's Theorem about the subword order.

We are now ready to prove that all antichains are finite.

\begin{theorem}
\label{th:antichain}
  Let $L\subset\Sigma^*$ be an antichain for $\prec$. Then $L$ is finite.
\end{theorem}
  
We prove this theorem  by using a technique similar to the variation  by
\cite{nashwilliams} of the proof of \cite{higman}.
  First, we need the following property:\\

\begin{lemma}
  Let $X$ be an antichain of $\Sigma^*$. If
  $\prec$ is a well quasi-order on $\textsc{Down}_\prec(X)\backslash X$
  then $X$ is finite.
\end{lemma}
\begin{proof} 
  We denote $Q=\textsc{Down}_\prec(X)\backslash X$, and suppose $Q$ is
  a well quasi-order for $\prec$. 
  Therefore the product and the associated product
  order $(\Sigma\times Q,\prec_\times)$ define also a well
  quasi-order. 
  We consider $A=\{(a,x) \mid a\in\Sigma x\in Q ax\in X\}$.
  Because $\prec$ is monotone, for all $(a,x),(a',x')\in A$,
  $(a,x)\prec_\times(b,y)\Rightarrow ax\prec by$. Indeed, in this case
  $a=b$ and $x\prec y\Rightarrow ax\prec ay.$
  So $A$ has to be an antichain of the well quasi-order $\Sigma\times
  Q$. Therefore $A$ is finite. By construction, this implies that $X$ is
  also finite.  
\end{proof}

\begin{proof}
  We can now end the proof of Theorem~\ref{th:antichain}. 
  Suppose we have an
  infinite antichain $X_0$. By applying recursively the previous
  lemma, there exists for all $i\in\N$,
  $X_{i+1}\subset\textsc{Down}_\prec(X_i)\backslash X_i$ that 
  is also an infinite antichain of $\Sigma^*.$ By definition of
  $\textsc{Down}_\prec$, for all $x\in X_{i+1}$, there exists $y\in
  X_i$ such that $x\prec y$, ie $x\subseteq y$. It is also possible to
  choose the elements $x$ such that no pair is sharing a common $y$.
  So $X_{i+1}\subseteq_{\mathcal P} X_{i}$, 
  and we have a infinite descending chain of $(\mathcal
  P(\Sigma^*),\subseteq_{\mathcal P})$. This would contradict
  Lemma~\ref{higmans}.
\end{proof}


From Propositions \ref{prop:trans}, \ref{prop:mon}, \ref{prop:wfund} and Theorem 
\ref{th:antichain}
it follows that
$(\Sigma^*,\prec)$ is a well quasi-order, completing the proof of 
Proposition \ref{prop:wqo}; thus, $L_{wait}(\G)$ is a regular language
for any TVG $\G$, concluding the proof of Theorem \ref{theo:wait}.
That is ${\cal L}_{wait}  = {\cal R}$.

\begin{remark}
  The reader should note that $\prec$ does not correspond to
  $\subseteq_{\mathcal P}$ if we were to identify a word $x$ with the
  subset of corresponding walks.
  Indeed, if $W$ denotes the set of walks, then 
  for $A,B\in\mathcal P(W)$, $A\subset B \Rightarrow A\subseteq_P
  B$, however if $A=\Gamma(x), B=\Gamma(y)$, then $A\subset B \Rightarrow
  y\prec x.$  
  Therefore  the above theorem cannot be derived by a simple
  application of the results of Higman.
\end{remark}

\section{Bounded waiting allowed}
\label{sec:bounded-waiting}

To better understand the power of waiting, 
we now turn our attention to  {\em bounded waiting}; that is when 
indirect journeys are considered feasible if the pause 
between consecutive edges has a bounded duration $d>0$.
We examine  the set  ${\cal L}_{wait[d]}$ of all languages expressed by time-varying graphs when waiting is allowed up to  $d$ time units, and prove the negative result that for any fixed $d\geq 0$, ${\cal L}_{wait[d]} = {\cal L}_{nowait}$. 
That is, the complexity of the environment is not affected by allowing waiting for a limited amount of time. 

The basic idea is to reuse the same technique as in Section~\ref{sec:no-waiting}, but with a dilatation of time, i.e., given the bound $d$, the edge schedule is time-expanded by a factor $d$ (and thus no new choice of transition is created compared to the no-waiting case).

\begin{theorem}
  For any duration $d$, ${\cal L}_{wait[d]}={\cal L}_{wait[0]}$ ({\it i.e.,} ${\cal L}_{nowait}$)
\end{theorem}
\begin{proof}

  Let $L$ be an arbitrary computable  language defined over a finite alphabet $\Sigma$.
  Let $d\in\N$ be the maximal waiting duration.
  We consider a TVG $\G_{2,d}$ structurally equivalent to $\G_2$ (see Figure~\ref{fig:nowaiting} in Section~\ref{sec:no-waiting}), {\it i.e.,},
  $\G_{2,d}= \TVG$ such that $V=\{v_0,v_1,v_2\}$, $E=\{\{ (v_0,v_1,i),i\in\Sigma\} \cup \{\{ (v_0,v_2,i),i\in\Sigma\}, \cup \{(v_1,v_1,i),i\in\Sigma\}\cup\{(v_1,v_2,i),i\in\Sigma\}\cup\{(v_2,v_1,i),i\in\Sigma\}\cup\{(v_2,v_2,i),i\in\Sigma\}\}$. 
  The initial state is $v_0$, and the accepting state is $v_1$.
  If $\varepsilon\in L$ then $v_0$ is also accepting.

  Based on the mapping $\varphi$ defined for $\G_2$ in Section~\ref{sec:no-waiting}, we define
  another mapping $\varphi_d$ that associates to any word $w$ the value
  $(d+1)\varphi(w).$
  We also define $\psi_d(t)$
  to be equal to $\varphi^{-1}(\lfloor\frac{t}{d+1}\rfloor)$ when it
  is defined.
  For instance, $\varphi_5(0110)$ in base $2$ gives $(101+1)\times 10110$ (i.e., $132$ in base 10).
  Reversely, we have $\psi_5(132)=...=\psi_5(137)=0110$, and $\psi_5(138)=...=\psi_5(143)=0111$.

  The presence and latency functions are now
  defined along the lines as those of $\G_2$, the only
  difference being that we are using $\varphi_d$ (resp. $\psi_d$) instead
  of $\varphi$ (resp. $\varphi^{-1}$).
  Thus, for all $u\in\{v_0,v_1,v_2\}$,  $i\in\Sigma$, and $t\geq 0$, we  define 
  \begin{itemize}
  \item
    $\rho((u,v_1,i),t) = 1$ iff $\lfloor\frac{t}{d+1}\rfloor\in\varphi_d(\Sigma^*)$ and $\psi_d(t).i\in L,$
  \item  $\zeta((u,v_1,i),t) = \varphi_d(\psi_d(t).i)-t$
  \item $\rho((u,v_2,i),t) = 1$ iff $\lfloor\frac{t}{d+1}\rfloor\in\varphi_d(\Sigma^*)$ and $\psi_d(t).i\notin L,$
  \item $\zeta((u,v_2,i),t) = \varphi_d(\psi_d^{-1}(t).i)-t$
  \end{itemize}

  By the same induction technique as in Section~\ref{sec:no-waiting},
  we have that $L \subseteq L(\G_{2,d})$. Similarly, we have that any journey labeled by $w$
  ends at time exactly $\varphi_d(w)$, even if some
  $d-$waiting occurred. 

  Finally, we remark that for all words $w,w'\in\Sigma^+$ such that
  $w\neq w'$, 
  we have $|\varphi_d(w)-\varphi_d(w')|>d$. 
  Indeed, if $w\neq w'$ then they differ by at least one letter. The
  minimal time difference is when this is the last letter and these
  last letters are $i,i+1$ w.l.o.g. 
  In this case, $|\varphi_d(w)-\varphi_d(w')|\geq d+1$ by definition of
  $\varphi_d$.
  Therefore waiting for a duration of $d$ does not enable more
  transitions in terms of labeling.
\end{proof}


\end{document}